\documentclass[11pt]{article}
\usepackage[numbers]{natbib}
\usepackage{fullpage}
\usepackage{microtype}
\usepackage{graphicx}
\usepackage{hyperref}  
\usepackage{subcaption}
\usepackage{booktabs} 

\usepackage{amsmath}
\usepackage{amssymb}
\usepackage{mathtools}
\usepackage{amsthm}
\usepackage{float}
\usepackage{tikz}

\usetikzlibrary{arrows.meta,shapes}

\usepackage[capitalize,noabbrev]{cleveref}

\theoremstyle{plain}
\newtheorem{theorem}{Theorem}[section]
\newtheorem{fact}[theorem]{Fact}

\newtheorem{lemma}[theorem]{Lemma}
\newtheorem{corollary}[theorem]{Corollary}
\theoremstyle{definition}

\theoremstyle{remark}

\usepackage[textsize=tiny]{todonotes}

\title{Learning-Augmented Binary Search Trees}

\author{ Honghao Lin \footnote{Equal Contribution.} \footnote{Computer Science Department, Carnegie Mellon University. \texttt{honghaol@andrew.cmu.edu.}}
	\and
	Tian Luo\footnote{Computer Science Department,
	Carnegie Mellon University.
	\texttt{tianl1@andrew.cmu.edu.}}
	\and
	David P. Woodruff\footnote{Computer Science Department, Carnegie Mellon University. \texttt{dwoodruf@cs.cmu.edu.} Honghao Lin and David Woodruff would like to thank for partial support from the National Science Foundation (NSF) under Grant No. CCF-1815840.}
}

\begin{document}

\date{\vspace{-5ex}}

\maketitle

\begin{abstract}
A treap is a classic randomized binary search tree data structure that is easy to implement and supports $O(\log n)$ expected time access. However, classic treaps do not take advantage of the input distribution or patterns in the input. Given recent advances in algorithms with predictions, we propose pairing treaps with machine advice to form a learning-augmented treap. We are the first to propose a learning-augmented data structure that supports binary search tree operations such as range-query and successor functionalities. With the assumption that we have access to advice from a frequency estimation oracle, we assign learned priorities to the nodes to better improve the treap's structure. We theoretically analyze the learning-augmented treap's performance under various input distributions and show that under those circumstances, our learning-augmented treap has stronger guarantees than classic treaps and other classic tree-based data structures. Further, we experimentally evaluate our learned treap on synthetic datasets and demonstrate a performance advantage over other search tree data structures. We also present experiments on real world datasets with known frequency estimation oracles and show improvements as well.
\end{abstract}

\section{Introduction}

Querying ordered elements is a fundamental problem in computer science. Classically, hash tables and various tree-based data structures such as Red-Black Trees, AVL Trees, and Treaps have been used to efficiently solve this problem. While hash tables are very efficient and extensively used in practice, trees are more memory-efficient and can dynamically resize with greater convenience. Additionally, search trees can offer extra functionality over hash tables in the form of successor/predecessor, minimum/maximum, order statistics, and range query capabilities. In practice, self-balancing binary search trees are used in routing tables, database indexing (B-Trees), the Linux kernel (Red-Black Trees), and various implementations of collections, sets, and dictionaries in standard libraries.

Classic binary search tree data structures often support these functionalities in $O(\log n)$ time. However, most binary search tree implementations, such as Red-Black Trees and AVL Trees, do not leverage patterns in the data to improve performance; instead, they provide a worst-case of $O(\log n)$ time per access. Splay Trees are able to implicitly take advantage of the underlying input distribution without any information about the distribution as they are, up to a constant factor, statically optimal and conjectured to be dynamically optimal \cite{splay}. Unfortunately, each access is accompanied by a series of rotations that is proportional to the number of nodes visited during the access, which increases the access time by a possibly large constant factor. On the other hand, if the underlying distribution is known, then one can generate a statically optimal tree in $O(n^2)$ time \cite{optbst} or an approximately statically optimal tree in $O(n\log n)$ time \cite{nearoptbst}; however, these methods do not allow for dynamic insertion and deletion operations.

A natural idea that arises is to use patterns in data to improve the efficiency of these data structures. In recent years, the field of learning-augmented algorithms has blossomed (see \cite{mitzenmacher2020algorithms} for a survey). Given a predictor that can output useful properties of the dataset, we can then leverage these predictions to optimize the performance of the algorithm based on the predicted patterns. In this paper, we present a learning-augmented binary search tree that is the first to support efficient range queries and order statistics. 

In summary, we present the following contributions:
\begin{itemize}
\item We introduce a learning-augmented Treap structure which exploits a rank prediction oracle to decrease the number of comparisons needed to find an element.
\item We analyze our learning-augmented Treap and provide theoretical guarantees for various distributions. We further show that our learning-augmented Treap is robust under a reasonable amount of noise from the oracle and that it performs no worse than a random Treap when the oracle is inaccurate, for any input distribution $\mathcal{D}$, up to an additive constant.
\item We experimentally evaluate the performance of our learning-augmented Treap over synthetic distributions and real world datasets. On synthetic distributions, we show improvements of over $25\%$ compared to the best classical binary search trees. On real world datasets, we show that the performance is comparable to the best among popular classical binary search trees and show significant improvements over non-learned Treaps.
\end{itemize}

\subsection{Motivation for Learning-Augmented Treaps}

In a binary search tree, more frequently accessed items should appear closer to the root to reduce the average number of comparisons in order to retrieve the item. However, Red-Black Trees, AVL Trees, and non-learned Treaps do not take advantage of this property, while Splay Trees exploit this only to the extent that more recent items are placed near the root. 

Given an oracle that predicts the ranks of elements, it is natural to build a tree in which the top-ranked items are near the root. These oracles are indeed realistic; for example, we can use frequency estimators to approximately rank the elements. Hashing-based approaches, such as Count-Min \cite{countmin} and Count-Sketch \cite{countsketch}, have been shown to be efficient and effective in this regard. Further, recent advances in the learning-augmented space have spurred the development of learning-augmented frequency estimators \cite{hsu}. In our experiments, we use the trained machine learning oracle from Hsu et al. \cite{hsu} as our frequency estimator.

With the availability of such a predictor, the motivation of augmenting a classic binary search tree data structure with a learned oracle is clear. Red-Black Trees and AVL Trees are uniquely determined by the insertion order of the elements and while it is feasible to insert the elements in such an order such that the top-ranked items are near the root, it is not clear how to support insertions and deletions to maintain this property. On the other hand, the order of insertions matters less for a Splay Tree, especially over a long sequence of operations, as it is self-adjusting. Our attempts at producing a learning-augmented Splay Tree have been unfruitful; this was largely due to the high overhead associated with rotations and difficulties in determining whether to splay an element. Instead of a Splay Tree, statically optimal trees could also be built with a frequency estimation oracle; however, these trees are unable to support insertions or deletions after initialization.

Treaps are simpler to analyze and can naturally be adapted in the learning-augmented setting. Indeed, Treaps are uniquely determined by the priorities of each key (given that all priorities are unique) and elements with higher priority appear closer to the root of tree. In this paper, we suggest assigning learned priorities to the Treap instead of priorities being drawn from a distribution $\mathcal{D}$; specifically, we assign the learned frequency as the priority. With this adjustment, we are able to efficiently support insertions and deletions, among other tree operations, while improving access time. We note that in the paper that introduced Treaps \cite{treaps}, a modification was suggested where every time an element was accessed, the new priority of the element was set to be the maximum of the old priority and a new draw from $\mathcal{D}$. We are the first to learn the priorities using a machine learning oracle.  

\subsection{Related Work}

This paper follows the long line of research in the growing field of algorithms with predictions. Learning-augmented algorithms have been applied to a variety of online algorithms, such as the ski rental problem and job scheduling \cite{ski}. Further, caches \cite{cache}, Bloom filters \cite{bloom}, index structures \cite{index}, and Count-Min and Count-Sketch \cite{hsu} are among the many data structures that have had a learning-augmented counterpart suggested. 

In particular, \cite{index} suggests replacing B-Trees (or any other index structure) with trained models for querying databases. In their work, instead of traversing the B-Tree to access the position of a record, they use a neural net to directly point towards the record. Our work is different from theirs since it keeps the search tree and instead optimizes the structure of the tree for faster queries; through this, we are able to support common additional tree-based operations such as traversal, order statistics, merge/join, etc.

Our work uses the frequency estimation oracle trained in Hsu et al. \cite{hsu} on a search query dataset (AOL) and an IP traffic dataset. Since then, other papers have used these predictions as the basis of their learned structures \cite{datastream}. Furthermore,`~\cite{betteroracle} has shown an improved oracle for the IP dataset, which shows significant improvements in accuracy.

\section{Preliminaries}

We use $[n]$ to denote the set $\{1, \dots, n\}$; further, we identify the set of keys in our binary tree with $[n]$. For a sequence of $m$ queries, we let $e_i \in [n]$ be the $i^{th}$ most frequent item with frequency $f_i$, breaking ties randomly. In our analysis, we also assume that the input distribution is the same throughout the duration of the query sequence and that the frequency observed in the sequence of $m$ queries exactly reflects its true distribution, as in element $e_i$ occurs with probability $p_i = \frac{f_i}{m}$. We will define the rank of element $e_i$ to be $i$ and the ranking ordering to be $e_1, \dots, e_n$.

\textbf{Treaps:} Treaps are binary search trees that also hold an additional field per node that stores the priority of that node. For node $i$, we denote the priority of the node to be $w_i$. At the end of any operation, in addition to the binary search tree order invariant, a Treap always satisfies the heap invariant, that is, if node $x$ is an ancestor of node $y$, then $w_x \geq w_y$. Classically, the priorities of a Treap are drawn from a continuous distribution so as to not have any duplicate priorities.

Insertion in a Treap is simple. We insert the new node by attaching it as a leaf in the appropriate position (i.e., satisfying the order invariant) in the Treap. Then, we repeatedly  rotate the new node and its parent until the heap invariant is satisfied. Deletion is achieved by rotating the node down until it is a leaf and detaching the node; we pick the child with the greater priority to perform the rotation.

We will refer to a Treap where priorities are assigned based on the rank or frequency of the element as a learned Treap, while a classic random Treap will be referred to as a random Treap.

Throughout the analysis, we make the assumption that the rank order of the keys is random, as in, $e_1, \dots, e_n$ is a random permutation of $[n]$. \hyperref[sec:var]{Section \ref*{sec:var}} shows how to remove this assumption; however, we keep this assumption for the first part of the analysis for ease. Furthermore, we will also assume for convenience that the frequencies of the keys are unique, as in $f_i = f_j$ if and only if $i = j$. To remove this assumption, we can break ties randomly. If elements $x$ and $y$ have the same frequency and the tie is broken in favour of $x$, we will say that $x$ has lower rank than $y$.

\textbf{Zipfian Distribution:} In our analysis, we analyze the performance of a Treap under the Zipfian distribution. Specifically, the Zipfian distribution with parameter $\alpha$ has frequencies $f_i = \frac{m}{i^\alpha H_{n, \alpha}}$ where $H_{n, \alpha} = \sum_{i = 1}^n \frac{1}{i^\alpha}$ is the $n^{th}$ generalized harmonic number of order $\alpha$.

\section{Learning Augmented Treaps}
\label{sec:analysis}

In the following sections, we assume access to a perfect oracle and analyze the theoretical performance of our learned Treaps versus random Treaps. In \hyperref[sec:robust]{Section \ref*{sec:robust}}, we discuss the robustness and performance guarantees of our learned Treaps when we are given a noisy oracle. In \hyperref[sec:lim]{Section \ref*{sec:lim}}, we explore performance guarantees of our oracle if given a less powerful oracle. Finally, in \hyperref[sec:var]{Section \ref*{sec:var}} we explore a modified version of the learned Treap that removes the assumption that the rank ordering is a random permutation.

\subsection{Treap Operations}

\subsubsection{Treap Initialization}

Given a predictor that outputs the frequency rank of an element, we assign a priority equal to the frequency rank of the element and insert the element into the Treap. Similarly, if we had a frequency estimation oracle instead of a rank-estimation oracle, we can insert the element into the Treap with priority equal to the frequency estimate. 

\subsubsection{Access}

We present the following theorem that bounds the expected depth of $e_i$ in a learned Treap.

\begin{theorem}
\label{thm:EDepth}
The expected depth of $e_i$ in a learned Treap is $2H_i-1$, where $H_i$ is the $i$-th Harmonic number.
\end{theorem}
\begin{proof} 
Consider the set of elements with higher priority, i.e., $S = \{e_k\mid k \leq i\}$. Notice that only elements in $S$ can be ancestors of $e_i$ and $e_i$ cannot be an ancestor of any element in $S$. Since $e_1, \dots, e_n$ is a random permutation of $[n]$, $e_1, \dots, e_i$ is a random permutation of $\{e_1,\dots,e_i\}$. 

Under these assumptions, the number of comparisons needed to access $e_i$ is equivalent to the number of comparisons needed to correctly insert a random element $x \in [i]$ in a sorted array of elements $[i]\setminus \{x\}$ using binary search, where pivots are chosen randomly. Here, the pivots are analogous to the ancestors of $e_i$.

This motivates the following recurrence for computing the expected depth of $e_i$:

$$T(i) = 
\begin{cases}
1 & {i = 1}\\
1 + \frac{2}{i-1}\sum_{k = 1}^{i-1}\left(\frac{k}{i}T(k)\right) & \text{otherwise}
\end{cases}
$$
which simplifies to
$$T(i) = 
\begin{cases}
1 & {i = 1}\\
\frac{2}{i} + T(i-1) & \text{otherwise}
\end{cases}
$$
This recurrence evaluates to $T(i) = 2H_i - 1$.
\end{proof}

\begin{theorem}
\label{thm:EDepthwhp}
The expected depth of $e_i$ in a learned Treap is $O(\log i)$ with high probability.
\end{theorem}
\begin{proof}
Again, we analyze the depth of element $e_i$ by examining the number of comparisons needed to insert a random element $x \in [i]$ in a sorted array of $[i]\setminus x$ use random binary search. We employ a similar technique to the classic high probability analysis of QuickSort.

Suppose on iteration $k$, the size of the array being searched is $X_k$. With probability $\frac 12$, the randomized pivot is situated in the range $[\frac 14 X_k, \frac 34 X_k]$. In this case, $X_{k+1} \leq \frac 34 X_k$. Otherwise, if the pivot does not land in this range, we know that $X_{k+1} \leq X_k$ trivially. We get the following:
$$\mathbb{E}[X_{k}] \leq \frac 12 X_{k-1} + \frac 38 X_{k-1} = \frac 78 X_{k-1}$$
Since $X_0 \leq i$, we have
$$\mathbb{E}[X_k] \leq \left(\frac{7}{8}\right)^k X_0 \leq \left(\frac{7}{8}\right)^ki$$
The probability that the randomized binary search uses more than $k$ iterations is exactly $\Pr\{X_k\geq 1\}$. Using Markov's Inequality and setting $k = c\log_{\frac{8}{7}}i$ for some constant $c$ gives

$$ 
\Pr\{X_k\geq 1\} \leq \mathbb{E}[X_k] \leq \frac{1}{i^{c-1}}
$$

Therefore, setting $c \geq 2$ implies that the expected depth of $e_i$ is $O(\log i)$ with high probability.
\end{proof}
\begin{theorem}
\label{thm:EntireTreeWHP}
With constant probability, for every $i$, $e_i$ has depth $O(\log i)$. In other words, the entire tree is well-balanced.
\end{theorem}
\begin{proof}
Again, let $X_i$ be the depth of element $e_i$. Notice that $X_1$ and $X_2$ are $1$ and $2$, respectively, with probability $1$. From \autoref{thm:EDepth}, for $i \geq 3$, $X_i \leq O(\log i)$ with probability at most $\frac{1}{i^{c-1}}$ for some constant $c$. Applying a union bound over elements $X_i$ for $i \geq 3$ gives

$$\Pr\{\bigcup_{i = 3}^n X_i \leq O(\log i)\} \leq \sum_{i = 3}^n \frac{1}{i^{c-1}}$$

For $c = 3$, $\sum_{i = 3}^n \frac{1}{i^{c-1}} \leq \frac{\pi^2}{6}-1.25 \approx 0.39$.
\end{proof}

\subsubsection{Insertion/Deletion and Priority Updates}

\begin{corollary}
\label{col:insertion}
The expected time of an insertion, deletion or priority update is $O(\log n)$.
\end{corollary}
\begin{proof}
Suppose during the insertion process of a node $x$, we attach $x$ to node $y$ as a leaf. Then by \hyperref[thm:EDepth]{Theorem \ref*{thm:EDepth}}, the depth of $y$ is $O(\log n)$ in expectation. Similarly, suppose during the deletion process of node $x$, we detach node $x$ when it is a child of node $y$. By \hyperref[thm:EDepth]{Theorem \ref*{thm:EDepth}}, the time it takes is $O(\log n)$. 

Similarly, a priority update takes at most $O(\log n)$ time. 
\end{proof}

\subsubsection{Other Operations}
\label{sec:otherops}
Our learned Treap could also be optimized for other tree-based operations. Under these modifications, the following operations could be supported by a learned Treap with time similar to that of an access on a learned Treap.

\textbf{Range Queries:} Consider the simple operation of counting the number of elements between keys $x$ and $y$. This operation is commonly implemented by augmenting every node with a field that stores the size of the subtree rooted at the node. Counting the number of elements in the ranges reduces to traversing from the root to $x$ and from the root to $y$, which is similar to the process of accessing $x$ and $y$. Thus, the predictor can learn the frequency distribution of the boundaries of the range query to optimize our Treap. Other such operations may include the standard RangeSum operation, which outputs the sum of the values stored in each key of the tree.

\textbf{Successor/Predecessor:} On a query to the successor of key $x$, the output would be the smallest key greater than $x$. Among the many ways to implement this functionality, a simple way is to keep a pointer that points to the successor/predecessor. When finding the successor of key $x$, we simply access $x$ in the Treap and use the stored pointer to access the successor. Our predictor can learn the frequency distribution of successor queries to $x$ and set the priority of $x$ accordingly in the learned Treap.

When supporting this operation, insertion and deletion become more complicated. When inserting element $x$, we must change the pointers of both the successor and predecessor of $x$ accordingly; however, this requires at most a constant number of pointer changes.


\subsection{General Distributions}

In this section, we analyze the expected cost per access of a learned Treap and a random Treap for an arbitrary frequency distribution $\mathcal{D}$.

\begin{lemma}
\label{lem:ECost}
The expected cost of a single access on a learned Treap is $\sum_{i = 1}^n p_i(2H_i - 1)$. 
\end{lemma}
\begin{proof}
This follows immediately from \hyperref[thm:EDepth]{Theorem \ref*{thm:EDepth}}.
\end{proof}

Since the expected cost of a single access is known to be at most $O(\log n)$ \cite{treaps}, we provide a lower bound on this expectation.
\begin{theorem}
\label{thm:ERanCost}
The expected cost of a single access on a random Treap is at least $2H_{n+1} - 4$ for any frequency distribution. 
\end{theorem}
\begin{proof}
The expected depth of key $i$ is well-known to be $H_i + H_{n-i+1} - 2$ \cite{treaps}. Let $X$ be the depth of an access and let $X_{ij}$ be the depth of key $i$ if it is the $j^{th}$-ranked item, and $0$ otherwise.
\allowdisplaybreaks
\begin{align*}
    \mathbb{E}[X] &=\sum_{i = 1}^n \frac{1}{n}\sum_{j = 1}^n  \mathbb{E}[X_{ij}] \\
    &=  \sum_{i = 1}^n \frac{1}{n}\sum_{j = 1}^n p_j(H_i + H_{n-i+1} - 2) \\
    &= \sum_{i = 1}^n \frac{1}{n}(H_i + H_{n-i+1} - 2) \\
    &= \frac{2}{n}\sum_{i = 1}^n H_i - 1 \\
    &= \frac{2}{n}((n+1)H_{n+1} - 2n)\\
    &> 2H_{n+1} - 4 \qedhere
\end{align*}

\end{proof}

\subsection{Zipfian Distributions}

In this section, we analyze the expected cost of an access of a learned Treap where $p_i \propto \frac{1}{i^\alpha}$ for a parameter $\alpha$.

\begin{theorem}
\label{thm:zipfgeneralalpha}
The expected cost of a single access on a learned Treap is $\sum_{i = 1}^n\frac{1}{i^\alpha H_{n, \alpha}}(2H_i -1)$.
\end{theorem}
\begin{proof}
From \hyperref[lem:ECost]{Lemma \ref*{lem:ECost}}, it is immediate that the expected cost is $\sum_{i = 1}^n\frac{1}{i^\alpha H_{n, \alpha}}(2H_i -1)$. 
\end{proof}

\begin{lemma}
\label{lem:zipfalpha1}
The expected cost of an access for $\alpha = 1$ is at most $H_n$. 
\end{lemma}
\begin{proof}
For $\alpha = 1$, the expected cost is $\sum_{i = 1}^n\frac{1}{i H_n}(2H_i -1)$ by \hyperref[thm:zipfgeneralalpha]{Theorem \ref*{thm:zipfgeneralalpha}}.

Consider the sum $C = \sum_{i = 1}^n\frac{1}{i}(H_i)$. Observe by expanding this summation that it evaluates to $\frac{1}{2}((H_n)^2 + H_{n, 2})$. The expected cost can then be expressed as $\frac{2C}{H_n} - 1 = H_n + \frac{H_{n, 2}}{2H_n} - 1.$ This approaches $H_n - 1$ as $n$ increases.
\end{proof}
\begin{corollary}
\label{col:2fac}
The expected cost of an access for a learned Treap on a Zipfian distribution with parameter $\alpha = 1$ is approximately a factor $2$ less than that of an access on a random Treap.
\end{corollary}

\begin{lemma}
\label{lem:zipfalphageq1}
The expected cost of an access for $\alpha > 1$ is constant. 
\end{lemma}      
\begin{proof}
For $\alpha > 1$, the expected cost is $\sum_{i = 1}^n\frac{1}{i^\alpha H_{n, \alpha}}(2H_i -1)$ by \hyperref[thm:zipfgeneralalpha]{Theorem \ref*{thm:zipfgeneralalpha}}.

Consider the series $a_i = \frac{H_i}{i^{\varepsilon}}$ and $b_i = \frac{1}{i^{\alpha-\varepsilon}}$ for some $\varepsilon > 0$. Observe the following properties:
\begin{itemize}
    \item Since $H_n \leq \ln(n) + 1$, $\lim_{n\to\infty} a_n = 0$. Further, $\{a_n\}$ is monotonically decreasing for large $n$.
    \item Since $\alpha > 1$, there exists $\varepsilon$ such that $\alpha - \varepsilon > 1$ and thus, $\sum_{i = 1}^n b_i \leq c$ for some constant $c$.
\end{itemize}

Recall Dirichlet's test: if $\{a_n\}$ is a monotonically decreasing sequence whose limit approaches $0$ and $\{b_n\}$ is a sequence such that $\sum_{i = 1}^\infty b_i$ is bounded by a constant, then $\sum_{i = 1}^\infty a_ib_i$ converges as well.

By these two observations and using Dirichlet's test, $\sum_{i = 1}^n a_nb_n = \sum_{i = 1}^n \frac{H_n}{n^\alpha}$ converges to a constant. The expected cost here is $\frac{2}{H_{n, \alpha}}\left(\sum_{i = 1}^n a_nb_n\right) - 1$. Therefore, the expected cost is bounded from above by a constant.
\end{proof} 
\begin{theorem}
\label{thm:staticopt}
The learned Treap is statically optimal in expectation for $\alpha \geq 1$.
\end{theorem}
\begin{proof}
First, consider $\alpha = 1$. The Shannon entropy, $H$, of this distribution is an asymptotic lower bound for a statically optimal tree, namely, Mehlhorn shows that for any binary tree, the weighted path length must be at least $\frac H3$ \cite{nearoptbst}. The Shannon entropy for the Zipfian distribution with $\alpha = 1$ is 
\begin{eqnarray*}
    && \sum_{i = 0}^n -p_i\log\left(p_i\right) 
     = \sum_{i = 0}^n\frac{1}{iH_n}\log\left(iH_n\right)\\
    & =& \sum_{i = 0}^n\frac{1}{iH_n}\left(\log\left(i\right)+\log(H_n)\right)
    \geq  \sum_{i = 0}^n\frac{1}{iH_n}\log\left(i\right)
\end{eqnarray*}
Clearly, this is within a constant factor of the expected cost of our learned Treap. Since the expected cost for the learned Treap is within a constant factor of the Shannon entropy, we are statically optimal up to a constant factor.

For $\alpha > 1$, the expected cost is constant; therefore, it is immediate that we are at most a constant factor more than the statically optimal binary search tree.
\end{proof}

\subsection{Noisy Oracles and Robustness to Errors}
\label{sec:robust}
In this section, we will prove that given an accurate rank prediction oracle subject to a reasonable amount of noise and error, our learned Treap's performance matches that of a perfect rank prediction oracle up to an additive constant per access.

Given element $i$, let $r_i$ be the real rank of $i$ and let $\hat{r}_i$ be the predicted rank of $i$. We will call an oracle noisy if $\hat{r}_i \leq \varepsilon r +\delta$ for some constants $\varepsilon, \delta \geq 1$.

\begin{theorem}
\label{thm:approxoracle}
Using predictions from a noisy oracle, the learned Treap's performance matches that of a learned Treap with a perfect oracle up to an additive constant.
\end{theorem}
\begin{proof}
The expected cost of a single access on the learned Treap with a noisy oracle is at most $\sum_{i = 1}^n p_i(2H_{\varepsilon i +\delta} - 1).$

The difference between the expected cost of a learned Treap with a noisy oracle and a learned Treap with a perfect oracle is 
$\sum_{i = 1}^n 2p_i(H_{\varepsilon i +\delta} - H_i).$

Using that $\ln(n) \leq H_n \leq \ln(n) + 1$ for the $n^{th}$ Harmonic number $H_n$, the difference is at most $\sum_{i = 1}^n 2p_i\left(1+\ln\left(\varepsilon+ \frac{\delta}{i}\right)\right) \leq \sum_{i = 1}^n 2p_i\left(1+\ln\left(\varepsilon+ \delta\right)\right) = 2\left(1+\ln\left(\varepsilon+ \delta\right)\right) \leq c,$ 
for some constant $c$. Therefore, under a noisy oracle, the learned Treap is at most an additive constant worse than a learned Treap with a perfect oracle.
\end{proof}
We remark that for frequency estimation oracles, it might be natural to consider an error bound of $\frac 1\Delta f_i \le \hat{f}_i \le \Delta f_i$ instead; however, if the underlying distribution is Zipfian, a frequency estimation error bound of $\frac 1\Delta f_i \le \hat{f}_i \le \Delta f_i$ is equivalent to a rank estimation error bound of $\hat{r}_i \in  r_i \pm \Delta^2$.

We will call an oracle inaccurate if there exist no constants $\varepsilon, \delta \geq 1$ such that $\hat{r}_i \leq \varepsilon r +\delta$. Further, we will define the notion of an adversarial oracle as an oracle that outputs a rank ordering that is adversarial; more specifically, given a distribution $\mathcal{D}$ with a random rank ordering, a non-adversarial oracle would output a random rank ordering that is not necessarily the same as the rank ordering of $\mathcal{D}$.

\begin{theorem}
A learned Treap based on an inaccurate but non-adversarial oracle has expected performance no worse than that of a random Treap, up to a small additive constant.
\end{theorem}
\begin{proof}
Since the oracle is non-adversarial, the expected depth of any element is still bounded by $2H_n - 1$ by \hyperref[thm:EDepth]{Theorem \ref*{thm:EDepth}}. Therefore, the expected cost is at most
$\sum_{i = 1}^n p_i(2H_n - 1) = 2H_n - 1.$ 
\end{proof}

\subsection{Oracles with Limited Capabilities}
\label{sec:lim}

In certain circumstances, it may be impossible or inconvenient to obtain an oracle that predicts the full rank ordering of the elements. Instead, it may be easier to obtain an oracle that predicts the top $k$ elements only.

In this case, we will assign the top $k$ elements random positive real-valued priorities and the remaining elements will be assigned random negative real-valued priorities. Hence, the top $k$ elements are ancestors of the remaining elements. Again, here, we will assume that the underlying rank ordering is a random permutation of $[n]$. Further, suppose that the top $k$ items account for $p$ percent of the queries.

\begin{theorem}
With an oracle that predicts only the top $k$ elements, the expected depth of an access is at most $2(pH_k + (1-p)H_n) - 1$.
\end{theorem}
\begin{proof}
For the top $k$ elements, the expected depth is at most $2H_k - 1$ and for the rest of the elements, the expected depth is at most $2H_n -1$. Therefore, the expected depth of an access is at most $2(pH_k + (1-p)H_n) - 1$.
\end{proof}

For small $k$ and significant $p$, this results in a large constant factor reduction in expected access depth.
Similarly, if we were given an oracle that can only accurately predict the frequencies of the top $k$ items, we can assign priorities of the top $k$ items to the frequency and assign random negative real-valued priorities to the remaining $n-k$ items.

\subsection{Removing Assumption of Random Rank Ordering}
\label{sec:var}

In real world datasets, it might not be the case that the rank ordering is a random permutation. For example, in search queries, certain queries are lexicographically close to misspelled versions of the query; however, misspelled versions of the query have a significantly reduced frequency compared to the correctly spelled query. Furthermore, it may be the case that the oracle is adversarial. In this case, we would want to remove the assumption that the rank ordering is a random permutation.

One natural idea is to map the identities of the elements to a random real number. For key $i$, we will use $s_i$ to denote this random real. The idea is to use a random Treap (or any other self-balancing binary search tree) and a learned Treap together. The random Treap will use the actual identity of the element as the key and the learned Treap will use the random real as the key. For each node in the learned Treap, we keep a pointer to the corresponding node in the random Treap. It immediately follows that the rank ordering on the keys of the learned Treap is equivalent to a random permutation. Furthermore, we keep a map that maps the identity of the element to its corresponding random real. We show an example of this modified learned treap in \autoref{badtikz}. 

\pgfdeclarelayer{background}
\pgfsetlayers{background,main}

\tikzstyle{random}=[circle,draw=black!100,fill=black!25,minimum size=22.5pt,inner sep=0pt]
\tikzstyle{learned} = [circle, draw=black!100, fill=black!0, minimum size=22.5pt, inner sep=0pt]
\tikzstyle{edge} = [draw,thick,-]
\tikzstyle{pointer} = [draw,thick,->,>=Latex,red!50]

\begin{figure}[ht]
\vskip 0.2in
\centering
\begin{tikzpicture}[scale=1.8, auto,swap]
    \foreach \pos/\name in {{(3.25,3)/5}, {(2.85,2.55)/2}, {(2.5,2.05)/1},{(3.15,2.05)/3}, {(3.65,2.55)/6}, {(3.3,1.55)/4}, {(3.9, 2.05)/7}}
        \node[random] (\name) at \pos {$\name$};
    \foreach \pos/\name in {{(1.05,3)/s_7}, {(0.65,2.55)/s_3}, {(0.3,2.05)/s_5},{(1.2,2.05)/s_2}, {(1.45,2.55)/s_4}, {(0,1.55)/s_1}, {(1.75, 2.05)/s_6}}
        \node[learned] (\name) at \pos {$\name$};    
    \foreach \source/ \dest in {5/2, 2/1, 2/3,5/6,6/7,3/4, s_7/s_3, s_3/s_5, s_1/s_5, s_7/s_4, s_4/s_6, s_2/s_4}
        \path[edge] (\source) -- (\dest);
    \begin{pgfonlayer}{background}
    \foreach \source/ \dest in {s_1/1, s_2/2,s_3/3,s_4/4,s_5/5,s_6/6,s_7/7}
        \path[pointer] (\source) -- (\dest);
    \end{pgfonlayer}
\end{tikzpicture}
\caption{An example of the learned Treap modification. White nodes form the learned Treap and grey nodes form the random Treap. The red arrows are the pointers from nodes in the learned Treap to the corresponding node in the random Treap. One possible assignment of $[s_1, \dots, s_7]$ for this Treap could be $[1,5,3,6,2,7,4]$.}
\label{badtikz}
\vskip -0.10in
\end{figure}
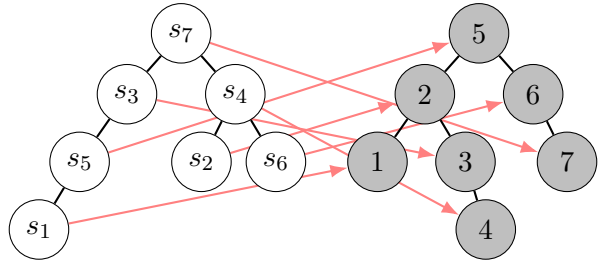

We describe each tree operation below:

\textbf{Access:} For an access operation to element $i$, we query $s_i$ in the learned Treap and use the pointer to access element $i$ in the random Treap.

\textbf{Insertion:} To insert element $i$, we generate $s_i$ and store $s_i$ in our map. Then we insert $i$ into the random Treap with a random priority and insert $s_i$ into the learned Treap with the learned priority. We set the pointer in the node containing $s_i$ to point to $i$.

\textbf{Deletions:} To delete element $i$, we delete $i$ from the random Treap, $s_i$ from the learned Treap, and remove $i$ and $s_i$ from the map.

\textbf{Successor/Predecessor:} To support successor and predecessor functionalities, we apply the same technique as described in \hyperref[sec:otherops]{Section \ref*{sec:otherops}} on the random Treap. 

Unfortunately, under this modification, there is no easy method of optimizing for range queries; however we note that this operation still takes at most $O(\log n)$ time in expectation because this is the expected sum of depths of the two nodes that we access. The main issue arises from the fact that range queries require access to the path from the root to the queried node on the random Treap; however, to remove the random rank ordering, we intentionally circumvent this path by traversing the learned Treap instead. For all accesses and successor/predecessor operations, we increase the cost of an operation by at most an additive constant related to accessing the map and a constant amount of pointer accesses. For insertions and deletions, we maintain the expected $O(\log n)$ bound since every node has expected depth at most $O(\log n)$. 

In practice, there might be a desire to avoid implementing a map; instead, using a hash function to implicitly store the map may be a more attractive alternative. We will show that using a $4$-wise independent hash function with range $(0, 1)$ would suffice. We choose to implement the hash function in $\mathrm{poly}(n)$ precision so that with high probability, there are no collisions and such a hash function requires $O(\log n)$ bits to store and only increases the cost of operations by at most an additive constant.

\begin{theorem}
\label{thm:limit_independence}
Given $s_{i} = h(e_i)$ where $h$ is drawn from a 4-wise universal hash family with range $(0, 1)$, the expected depth of $s_{i}$ is $O(\log i)$.
\end{theorem}

To achieve this, the following observation is crucial.
\begin{fact}
\label{fact:ancestor}
Suppose that $s_j$ is an ancestor of $s_i$ where $j < i$. Then in the ordering of $\{s_i | x \in \{1, \dots, j, i\}\}$, $s_i$ and $s_j$ are adjacent.
\end{fact}

\begin{proof}[Proof of Theorem~\ref{thm:limit_independence}]
Since the priorities of each key do not change, only elements in $\{e_1, \dots, e_{i-1}\}$ can potentially be ancestors of $e_i$. We proceed with an analysis similar to Knudsen and St\"{o}ckel's (\cite{quicksort}) analysis of quicksort under limited independence.

From Lemma 4 of \cite{quicksort}, we have the following lemma: given hash function $h: X \xrightarrow[]{} (0, 1)$ drawn from a $4$-universal hash family and disjoint sets $A,B \subseteq X$ with $|A| \leq |B|$, then $$\mathbb{E}[|\{a \in A | h(a) \leq \min_{b \in B} h(b)\}|] = O(1) \;. $$ Similarly, $\mathbb{E}[|\{a \in A | h(a) \geq \max_{b \in B} h(b)\}|] = O(1)$.

Consider the set $S_j = \{s_j | 1 \leq j \leq i-1\}$. From Fact~\ref{fact:ancestor} we get that if $s_j$ is an ancestor of $s_i$ for some $j < i$ , then for all $j' < j$, $s_{j'} < \min\{s_i, s_j\}$ or $s_{j'} > \max\{s_i, s_j\}$.


For $k = 1, 2, ..., \log i$, define $$B_k = [2^{k-1}] \text{ and } A_k = \left([2^k]\cap [i]\right)/[2^{k-1}] \; .$$

Suppose that $s_j$ is an ancestor of $s_i$ for some $j \in A_k$. Without loss of generality, we assume that $s_j < s_i$. Then we have that for each $j' \in B_k$, $s_{j'} < s_j$ or $s_{j'} > s_i$. Consider the hash function $H: X \xrightarrow[]{} (-(1-s_i), s_i)$ such that $H(x) = h(x)$ if $h(x) < s_i$ and $H(x) = h(x) - 1$ if $h(x) > s_i$. Notice that $H$ is also a $4$-wise independent hash function. This implies that  $H(j) > \max_{b \in B_k}H(b)$.
From the lemma above, there are an expected $O(1)$ such elements in $A_k$ and since there are only $O(\log i)$ values of $k$ for which $A_k$ is non-empty,
it follows immediately by linearity of expectation that the expected number of ancestors of $s_i$ is $O(\log i)$ and thus, the expected depth of $e_i$ is $O(\log i)$.
\end{proof}
\section{Experiments}

In this section, we present experimental results that compare the performance of our learned Treap to classical self-balancing binary search tree data structures. Specifically, we examined Red-Black Trees, AVL Trees, Splay Trees, B-Trees of order $3$, and random Treaps. For binary search trees sensitive to insertion order, we insert all keys in a random order. For these experiments, we only consider query operations and report the total number of comparisons made by each data structure. We note that although the number of comparisons is not a precise measurement of actual runtime, with the exception of Splay Trees, traversing the tree is extremely similar across all data structures, and for all data structures tested except B-Trees, the number of comparisons is exactly the access depth. For Splay Trees, we can expect a constant factor more in actual runtime due to the rotations involved. 

\subsection{Synthetic Datasets}

We consider synthetic datasets where elements appear according to a Zipfian distribution with parameter $\alpha$. As with \autoref{sec:analysis}, we assume that the rank order of the elements is a random permutation. For each experiment, we consider a sequence of length $10^5$.





\begin{figure*}
    \centering
    \begin{minipage}[t]{.31\textwidth}
    \includegraphics[width=\columnwidth]{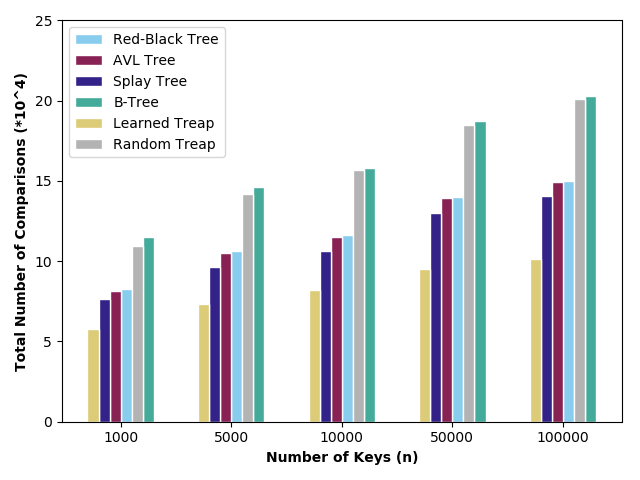}
    \vspace{-21pt}
    \caption{Total number of comparisons of classical binary search tree data structures and the learned Treap on the Zipfian Distribution with parameter $\alpha = 1$}
    \label{alpha1}
    \end{minipage}
    \hfil
    \begin{minipage}[t]{.31\textwidth}
    \includegraphics[width=\columnwidth]{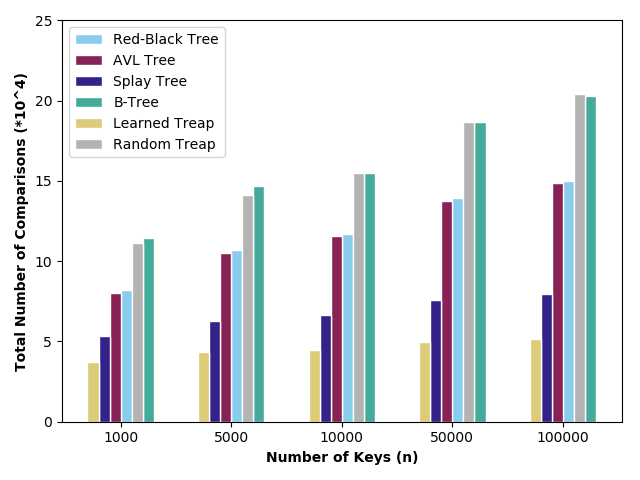}
    \vspace{-21pt}
    \caption{Total number of comparisons of classical binary search tree data structures and the learned Treap on the Zipfian Distribution with parameter $\alpha = 1.25$.}
    \label{alpha125}
    \end{minipage}
    \hfil
    \begin{minipage}[t]{.31\textwidth}
    \includegraphics[width=\columnwidth]{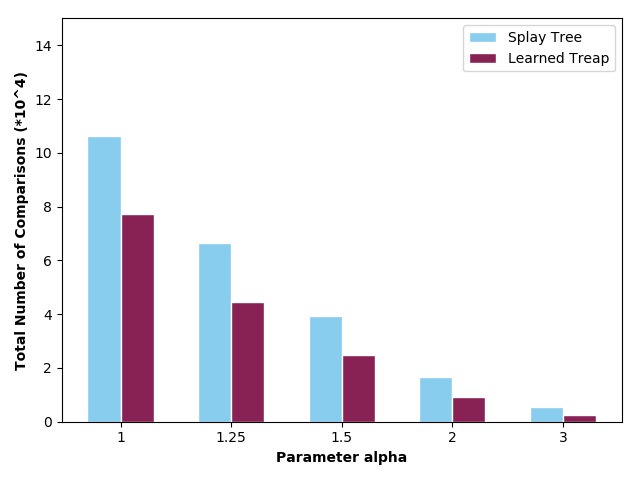}
    \vspace{-21pt}
    \caption{Total number of comparisons of Splay Tree and learned Treaps for varying Zipfian parameter $\alpha$.}
    \label{varyalpha}
    \end{minipage}
\end{figure*}

We report experimental results where we vary $n$, the number of keys, for $\alpha =1$ in \autoref{alpha1} and for $\alpha = 1.25$ in \autoref{alpha125}.

Notice that for both $\alpha = 1$ and $\alpha = 1.25$, the learned Treap performs approximately $25\%$ better than Splay Trees and a bit over $30\%$ better than AVL and Red-Black Trees in terms of the number of comparisons. For $\alpha = 1$, the factor-$2$ savings shown in \hyperref[col:2fac]{Corollary \ref*{col:2fac}} is exhibited and for $\alpha = 1.25$, we can see that the cost of an access is constant, as shown in \hyperref[lem:zipfalphageq1]{Lemma \ref*{lem:zipfalphageq1}}.

In \autoref{varyalpha}, we show the effects of varying $\alpha$; in this set of experiments, we fix the number of keys to be $10^4$ and only show results for the statically optimal trees, as in Splay Trees and learned Treaps. The learned Treap performs between approximately $27\%-51\%$ better than the Splay Tree. The greatest improvement was at $\alpha = 3$ and the least improvement was observed when $\alpha = 1$.

\subsection{Real World Datasets}

In this section, we used machine learning models trained by~\cite{hsu} as our frequency estimation oracle. We present 4 versions of our learned Treap. We consider the performance of our learned Treap with the trained frequency estimation oracle and with a perfect oracle; for both of these instances, we also test the performance if we remapped the keys to a random permutation (i.e., similar to the idea of \hyperref[sec:var]{Section \ref*{sec:var}}). We call the remapped versions of the learned Treap ``shuffled". To make the data more presentable, among classical binary search tree data structures, we only show the results of Red-Black Trees and Treaps; we remark that the relative performance of all classical binary search tree data structures in these datasets was similar to that in the synthetic datasets. 


\begin{figure}[!h]
    \centering
    \begin{minipage}[t]{.45\textwidth}
    \includegraphics[width=\columnwidth]{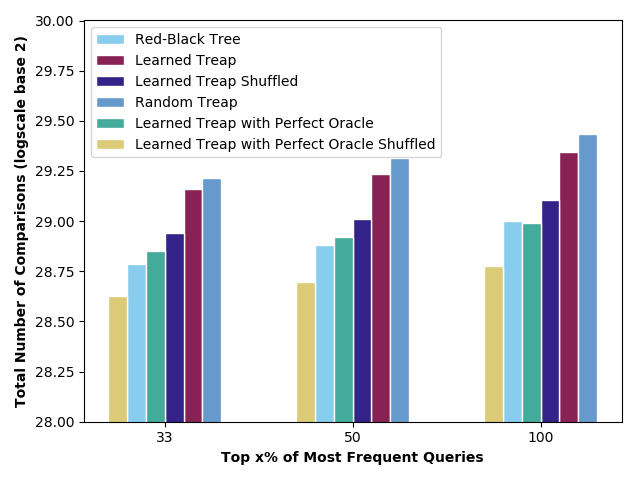}
    \vspace{-10pt}
    \caption{Total number of comparisons of Red-Black Trees, random Treaps, and learned Treaps on the $20^{th}$ test minute}
    \label{Internet}
    \end{minipage}
    \begin{minipage}[t]{.45\textwidth}
    \includegraphics[width=\columnwidth]{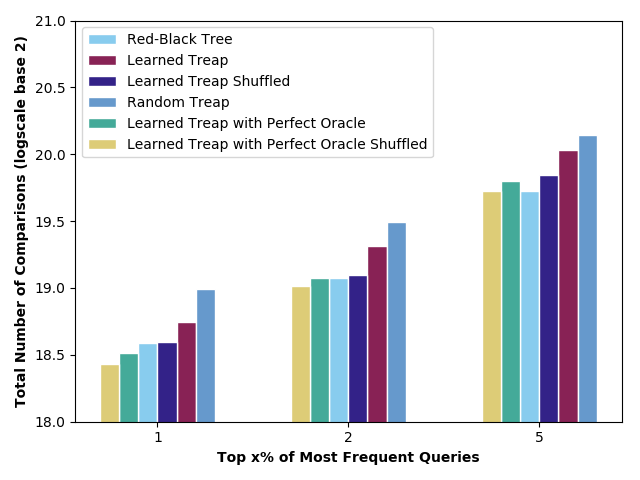}
    \vspace{-10pt}
    \caption{Total number of comparisons of Red-Black Trees, random Treaps, and learned Treaps on the $50^{th}$ day.}
\label{aol}
    \end{minipage}
\end{figure}

\subsubsection{Internet Traffic Data}

Various forms of self-balancing binary search trees and skip lists have been suggested to be used in routing tables \cite{unix}. In this experiment, we measure the performance of the binary search trees if we had to query every packet in the internet traffic logs.

\textbf{Data:} The internet traffic data was collected by CAIDA using a commercial backbone link (Tier 1 ISP) \cite{CAIDA}. Following~\cite{hsu}, we used the internet traffic recorded from Chicago outgoing to Seattle recorded on 2016-01-21 13:00-14:00 UTC. Each minute recorded approximately 30 million and 1 million unique flows. 

\textbf{Model:} We used the prediction made by~\cite{hsu}. In their paper, an RNN was used to encode the source and destination IP addresses, ports, and protocol type, and a separate RNN was used to predict the number of packets from the traffic flow based on the encoding. The first 7 minutes of the dataset was used as training sets with the next 2 minutes used as the validation sets. The rest of the dataset was used for testing. See Hsu et al. \cite{hsu} for details.

\textbf{Results:} In \autoref{Internet}, we plot the performances of the various data structures. We consider three variants of the dataset: a subset with the top 33\% of the most frequent queries, a subset of the top 50\% of the most frequent queries, and the full dataset. We show the results on the $20^{th}$ test minute (2016-01-21 13:29 UTC). 


In all cases, the shuffled versions of the learned Treap perform significantly better than that of the non-shuffled versions, and the learned Treaps perform better than random Treaps. We note that using the oracle from~\cite{hsu}, we are unable to beat Red-Black Trees; however, the shuffled learned Treap with the learned oracle is comparable and with a better oracle, it could be possible to outperform a Red-Black Tree. 

\subsubsection{Search Query Data}

\textbf{Data:} This dataset contains approximately 21 million queries on AOL collected over 90 days in 2006. The distribution follows Zipf's Law (see Hsu et al. \cite{hsu}).

\textbf{Model:} Again, we use the predictions from ~\cite{hsu}. They use an RNN with LSTM cells to encode the queries character by character. The encoding is then fed into a fully connected layer to predict the frequency of each query. The first 5 days were used for training while the $6^{th}$ day was used as the validation set.

\textbf{Results:} As with the Internet traffic dataset, we show the performance of the learned Treaps, a Red-Black Tree, and a random Treap in \autoref{aol}. For this dataset, we consider the top 1\%, 2\%, and 5\% of the most frequent queries as our set of keys. We show the results for the $50^{th}$ day.

Similar to the internet traffic dataset, the shuffled version of the learned Treaps performed better and all learned Treaps performed better than the random Treap. For this dataset, the shuffled learned Treap with the frequency estimator from~\cite{hsu} performed well and is comparable to the performance of a Red-Black Tree. Furthermore, unlike the internet traffic dataset, the performance of the learned Treaps with the machine learning model was close to that of the learned Treap with a perfect oracle.

\begin{figure*}[!h]
    \centering
    \includegraphics[width=0.32\columnwidth]{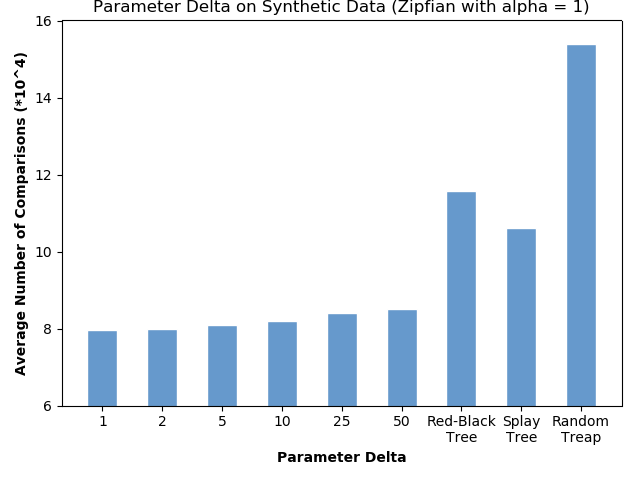}
    \includegraphics[width=0.32\columnwidth]{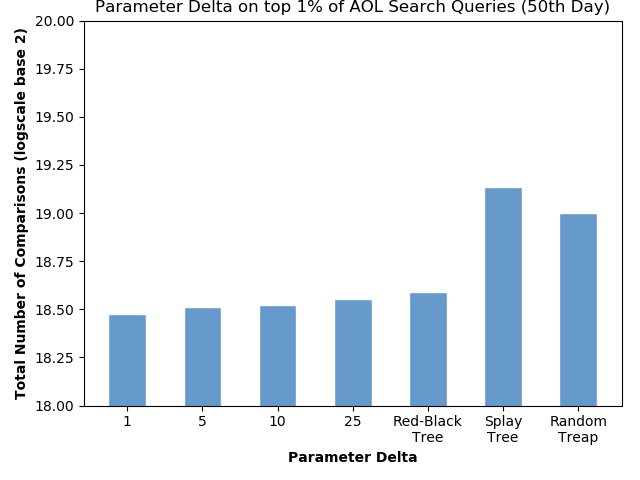}
    \includegraphics[width=0.32\columnwidth]{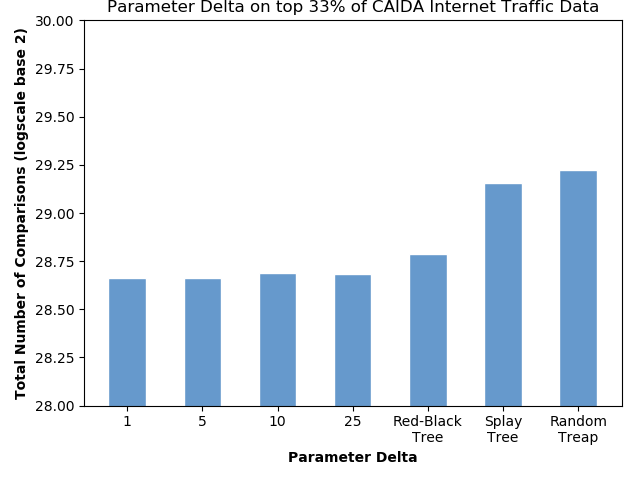}
    \vspace{-10pt}
    \caption{Performance of learned Treap under oracles with different errors}
    \label{fig:noisy_oracles}
\end{figure*}

\subsection{Performance under Oracles with Different Errors}

In this section, we study the performance of the learned Treap under oracles with certain errors on both synthetic and real-world data. 
In \autoref{fig:noisy_oracles} we show experimental results on synthetic and real-world data that show a graceful degradation as error grows. Here the prediction, $\hat{f}_i$, of the frequency satisfies $\hat{f}_i \le \Delta f_i$. 
We note that if the underlying distribution is Zipfian, then the error bounds for the rank-estimation oracle are stronger than the bounds for a frequency estimation oracle; if a given  frequency estimation oracle has the error bound of $\frac 1\Delta f_i \le \hat{f}_i \le \Delta f_i$, then under a Zipfian distribution with $\alpha \geq 1$, then $\hat{r}_i \in  r_i \pm \Delta^2$.

\section{Conclusion}

We introduced the concept of learning-augmented algorithms into the class of binary search tree data structures that support additional operations beyond B-trees. The learned Treap is able to support various useful tree-based operations, such as range-queries, successor/predecessor, and order statistic queries and can be optimized for such operations. We proved that the learned Treap is robust under rank-estimation oracles with reasonable error and under modifications, is no worse than a random Treap regardless of the accuracy of the oracle and the underlying input distribution. Further, we presented experimental evidence that suggests a learned Treap may be useful in practice. In the future, it may be interesting to explore whether advanced tree data structures, such as van Emde Boas Trees or Biased Skip Lists, can also benefit from machine learning techniques.

\bibliographystyle{alpha}
\bibliography{references}

\newpage

\appendix
\onecolumn

\end{document}